\documentclass[a4paper,11pt,english]{article}
\usepackage[utf8]{inputenc}
\usepackage{babel,amsmath,amssymb,amsthm,enumitem}
\usepackage[sort&compress,numbers,square,sort]{natbib}
\usepackage[colorlinks,allcolors=blue]{hyperref}
\usepackage[margin=1.25in]{geometry}
\usepackage[small,bf]{titlesec}

\titlelabel{\thetitle.\hspace{0.5em}}

\newtheorem{theorem}{Theorem}
\newtheorem{lemma}{Lemma}
\newtheorem{proposition}{Proposition}
\theoremstyle{definition}
\newtheorem{remark}{Remark}
\newtheorem{example}{Example}
\newtheorem{assumption}{Assumption}

\DeclareMathOperator{\E}{E}
\DeclareMathOperator{\Var}{Var}
\DeclareMathOperator{\I}{I}
\renewcommand{\P}{\mathrm{P}}
\newcommand{\F}{\mathcal{F}}
\newcommand{\FF}{\mathbb{F}}
\newcommand{\R}{\mathbb{R}}
\renewcommand{\tilde}{\widetilde}
\renewcommand{\epsilon}{\varepsilon}

\title{Asymptotic minimization of expected time to reach\\
a large wealth level in an asset market game}
\author{Mikhail Zhitlukhin\thanks{Steklov Mathematical Institute of the
Russian Academy of Sciences. 8 Gubkina St., Moscow, Russia. Email:
mikhailzh@mi-ras.ru. The research was supported by the Russian Science
Foundation, project no.\ 18-71-10097.}}
\date{9 July 2020}

\begin{document}
\maketitle

\begin{abstract}
We consider a stochastic game-theoretic model of a discrete-time asset
market with short-lived assets and endogenous asset prices. We prove that
the strategy which invests in the assets proportionally to their expected
relative payoffs asymptotically minimizes the expected time needed to reach
a large wealth level. The result is obtained under the assumption that the
relative asset payoffs and the growth rate of the total payoff during each
time period are independent and identically distributed.


\end{abstract}

\section{Introduction}
One of classical problems in mathematical finance consists in finding an
investment strategy which reaches a given wealth level as quick as possible.
It is generally known that in market models with exogenously specified asset
prices log-optimal strategies asymptotically minimize the expected time of
reaching a large wealth level, at least when asset returns are specified by
i.i.d.\ random variables; see e.g.\ the seminal paper \cite{Breiman61} for a
result in discrete time, or a more recent work \cite{KardarasPlaten10} for a
continuous-time model with L\'evy processes. In the present paper we obtain
an analogous result for a game-theoretic model of a market with endogenous
prices -- an \emph{asset market game}.

We consider a discrete-time model in which assets yield random payoffs that
are divided between agents (investors) proportionally to the number of
shares of each asset held by an investor. Asset prices are determined
endogenously by an equilibrium of supply and demand and depend on investors'
strategies. As a result, the evolution of investors' wealth depends not only
on their own strategies and realized asset payoffs but also on strategies of
the other investors in the market. Our goal is to identify an investment
strategy that allows an investor to reach a large wealth level
asymptotically not slower, on average, than any other investor in the
market. We show that there exists a strategy with this property, and,
moreover, it does not depend on the strategies used by the other investors.

This research was motivated by results in \emph{evolutionary finance} -- the
field which studies financial markets from a point of view of evolutionary
dynamics and investigates properties of investment strategies like survival,
extinction, dominance, and how they affect the distribution of wealth;
recent reviews of this direction can be found in
\cite{Holtfort19,EvstigneevHens+16}. We will work within the evolutionary
model of a market with \emph{short-lived} assets proposed by
\citet{AmirEvstigneev+13} (among earlier models of a similar structure one
can mention, e.g.,
\cite{BlumeEasley92,EvstigneevHens+02,HensSchenkHoppe05}). Short-lived
assets can be purchased by investors at time $t$, yield payoffs at $t + 1$,
and then the cycle repeats. They have no liquidation value, so investors can
get profit or loss only by receiving asset payoffs and paying for buying new
assets. Certainly, such a model is a simplification of a real stock market,
however models with short-lived assets have been widely studied in the
literature because they are more amenable to mathematical analysis and ideas
developed for them may be transferred to more realistic models.

The main results of the present paper are related to the strategy
$\lambda^*$ of \cite{AmirEvstigneev+13}, which splits an investment budget
between assets proportionally to their expected payoffs. In~that paper, it
was shown that $\lambda^*$ is a \emph{survival} strategy in the sense that
it allows investors using it to keep their relative wealth (the share in the
total market wealth) bounded away from zero on the whole infinite time
interval with probability 1. As observed in \cite{AmirEvstigneev+13}, in
view of such a structure, this strategy is analogous to the Kelly rule of
``betting one's beliefs'' in markets with exogenous asset prices (see
\cite{Kelly56}; a collection of papers on the Kelly rule can be found in
\cite{MacLeanThorpZiemba11}). Moreover, the key step to show that
$\lambda^*$ is a survival strategy was to prove that it makes the logarithm
of the relative wealth of an investor who uses it a submartingale, which is
analogous to the log-optimality property in markets with exogenous prices
(see, e.g., \cite{AlgoetCover88}, or later literature where such strategies
are often called growth-optimal, benchmark, or num\'eraire portfolios,
\cite{Long90,KaratzasKardaras07,PlatenHeath06}).

Our first main result shows that the expected time needed for an investor
using $\lambda^*$ to reach a wealth level $l$ is asymptotically, as
$l\to\infty$, not greater than the same time for any other investor in the
market, i.e.\ if we denote these times by $\tau_l^*$ and $\tau_l$, then
$\xi:=\limsup_{l\to\infty}\E \tau_l^*/ \E \tau_l \le 1$. Compared to
\cite{AmirEvstigneev+13}, where no conditions on the distribution of asset
payoffs are imposed, we require that the payoffs are generated by sequences
of i.i.d.\ random variables in a certain way, which is a usual assumption in
various settings of time minimization problems in asset market models (cf.,
e.g., \cite{Breiman61,KardarasPlaten10}), as well as in earlier works in
evolutionary finance.

The second main result states that, under some additional conditions, the
strict inequality $\xi<1$ takes place if the strategy of the other investor
is essentially different from $\lambda^*$ in some sense, and we find an
upper bound for $\xi$ which is strictly less than~1. It is interesting to
note that, among its assumptions, this results requires the payoffs to be
strictly random, and we provide a counterexample with non-random asset
payoffs where $\xi=1$. In other words, volatility, which we associate here
with randomness of payoffs, helps $\lambda^*$ to beat other strategies (see
further discussion in Section~\ref{section-results}).

The paper is organized as follows. Section~\ref{section-model} describes the
model, Section~\ref{section-results} states the main results, and
Section~\ref{section-proofs} contains their proofs.

\section{The model}
\label{section-model}
The model we use is essentially equivalent to that of
\cite{AmirEvstigneev+13}, but it will be more convenient to formulate it
using the notation which is more common in stochastic analysis.

For ease of exposition, let us first briefly describe the structure of the
model in plain language. The market consists of $M\ge 2$ agents (investors)
and $N\ge 2$ assets. At each moment of time $t=1,2,\ldots$\,, the assets
yield random payoffs, which are divided between the investors proportionally
to the number of shares of each asset purchased by an investor at time
$t-1$. The supply of each asset is given exogenously (and, without loss of
generality, is normalized to 1), while the demand depends on actions of the
investors, i.e.\ their investment strategies. An investment strategy
consists of an investor's decisions, made at every moment of time
simultaneously with the other investors and independently of them, on what
proportion of wealth to spend on buying each asset. Asset prices are
determined by means of the market clearing mechanism, i.e.\ they are set in
such a way that the demand becomes equal to the supply. Then, at the next
moment of time, the assets purchased by the investors yield payoffs and the
cycle repeats. The important simplifying modeling assumption consists in that
the assets bought at time $t-1$ cannot be sold at $t$, i.e.\ they disappear
after yielding payoffs without any liquidation value and are replaced by
their ``copies''. Hence, we can say that they live for just one period and
call them \emph{short-lived}.

To define the model formally, introduce a probability space $(\Omega,\F,\P)$
with a filtration $\FF=(\F_t)_{t=0}^\infty$. Payoffs of asset $n=1,\ldots,N$
are specified by a sequence of random variables $X_t^n\ge 0$, $t\ge 1$,
which is $\FF$-adapted (i.e.\ $X_t^n$ are $\F_t$-measurable). It is assumed
that $X_t^n$ are given exogenously, i.e.\ do not depend on actions of the
investors, and that $\sum_n X_t^n > 0$ for all $t\ge 1$ (hereinafter all
equalities and inequalities for random variables are assumed to hold with
probability 1).

The wealth of investor $m=1,\ldots,M$ is specified by an adapted random
sequence $Y_t^m\ge 0$. The initial wealth $Y_0^m$ of each investor is
non-random and strictly positive. Further evolution of wealth depends on the
investors' actions and the asset payoffs. Actions of investor $m$ are
represented by a sequence of vectors of investment proportions $\lambda_t^m
= (\lambda_t^{m,1},\ldots,\lambda_t^{m,N})$, ${t\ge 1}$, according to which
this investor allocates the available budget $Y_{t-1}^m$ for buying assets
at time $t-1$. Short sales are not possible and the whole wealth is
reinvested, so the vectors $\lambda_t^m$ belong to the standard $N$-simplex
$\Delta =\{\lambda\in \R_+^{N} : \sum_n \lambda^{n} =1\}$.

To allow dependence on a random outcome and the history of the market, we
define a strategy of an investor as a sequence of functions
\[
\Lambda_t(\omega, y_0,\lambda_1,\ldots,\lambda_{t-1})\colon \Omega\times
\R^M_+\times \Delta^{M(t-1)}\to \Delta,\qquad t\ge1,
\]
which are
$\F_{t-1}\otimes\mathcal{B}(\R^M_+\times\Delta^{M(t-1)})$--measurable
($\mathcal{B}$ stands for the Borel $\sigma$-algebra). The argument $y_0\in
\R_+^M$ corresponds to the vector of initial wealth
$Y_0=(Y_0^1,\ldots,Y_0^M)$. The arguments $\lambda_s=(\lambda_s^{m,n})$,
$m=1,\ldots,M$, $n=1,\ldots,N$, $s=1,\ldots,t-1$, are investment proportions
selected by the investors at the past moments of time (for $t=1$, the
function $\Lambda_1(\omega,y_0)$ does not depend on $\lambda_s$). If this
strategy is used by investor $m$, then the value of the function $\Lambda_t$
corresponds to the vector of investment proportions $\lambda_t^{m}$. The
measurability of $\Lambda_t$ in $\omega$ with respect to $\F_{t-1}$ means
that future payoffs are not known to the investors at a moment when they
decide upon their actions.

After the investors have chosen their investment proportions at time $t-1$,
the equilibrium asset prices $p_{t-1}^n$ are determined from the market
clearing condition that the aggregate demand of each asset is equal to the
aggregate supply, which is normalized to~1. Since investor $m$ can buy
$x_t^{m,n} = \lambda_t^{m,n} Y_{t-1}^m / p_{t-1}^n$ units of asset $n$, we
must have
\[
p_{t-1}^n = \sum_{m=1}^M \lambda_{t}^{m,n} Y_{t-1}^m.
\] 
If $\sum_m\lambda_t^{m,n}=0$, i.e.\ no one invests in asset $n$, we put
$x_t^{m,n}=0$ for all $m$; in this case the price $p_{t-1}^n=0$ can be
defined in an arbitrary way with no effect on the investors' wealth, so we
will put $p_{t-1}^n=0$ for convenience. At the next moment of time $t$, the
total payoff received by investor $m$ from the assets in the portfolio will
be equal to $\sum_n x_t^{m,n} X_t^n$. Consequently, the wealth sequence
$Y_t^m$ is defined by the recursive relation
\begin{equation}
Y_{t}^m(\omega) =  \sum_{n=1}^N
\frac{\lambda_t^{m,n}(\omega) Y_{t-1}^m(\omega)}{\sum_k
\lambda_{t}^{k,n}(\omega) Y_{t-1}^k(\omega)} X_t^n(\omega),\label{capital-equation-1}
\end{equation}
where $\lambda_t^{m,n}(\omega)$ denotes the \emph{realization} of
investor $m$'s strategy in this market, which is defined recursively as the
sequence
\begin{equation}
\lambda_t^{m,n}(\omega) = \Lambda_t^{m,n}(\omega, Y_0, \lambda_1(\omega),
\ldots,\lambda_{t-1}(\omega)).\label{realizations}
\end{equation}

Equation \eqref{capital-equation-1} expresses the wealth dynamics of an
individual investor in the market. Observe that $Y_t^m$ implicitly depends
on the strategies of the other investors. At the same time, if some investor
$m$ uses a fully diversified strategy, i.e.\
\begin{equation}
\label{positive-proportions}
\lambda_{t}^{m,n} > 0\ \text{for all}\ t,n,
\end{equation}
then $Y_t^m>0$ for all $t$ and the total market wealth, which we will denote
by $W_t = \sum_m Y_t^m$, does not depend on the investors' strategies and is
equal to $\sum_n X_t^n$.

\begin{remark}[On extensions of the model]
The main features of the model which considerably simplify its mathematical
analysis are that (a)~the assets are short-lived, (b)~the whole wealth is
reinvested and there is no risk-free asset, (c)~there are no short-sales,
and (d)~the time runs discretely. There is a number of papers where these
assumptions are relaxed. Among them, one can mention, for example,
\cite{AmirEvstigneev+11}, \cite{BelkovEvstigneev+20,DrokinZhitlukhin20},
\cite{AmirBelkov+20}, \cite{Zhitlukhin20} which address, respectively, the
limitations (a), (b), (c), (d). However, to my knowledge, there is no
general model which would combine all these extensions together.
\end{remark}

\section{Main results}
\label{section-results}
For a number $l>0$, let $\tau_l^m$ denote the stopping time when the wealth
of investor $m$ reaches or exceeds the level $l$ for the first time, i.e.
\[
\tau_l^m = \min\{t\ge 0: Y_t^m \ge l\},
\]
where $\min\emptyset = \infty$. We are interested in finding a strategy
which makes $\E \tau_l^m$ small compared to other strategies asymptotically
as $l\to\infty$.

Our first result, Theorem \ref{theorem-1} below, provides such a strategy in
an explicit form. We will prove it under the following assumption on the
payoff sequences.

\begin{assumption}
\label{A}
 The sequences $X_t^n$  can be represented in the form
\begin{equation}
\label{main-assumption}
X_1^n = \rho_1 R_1^n \sum_{m=1}^M Y_0^m, \qquad
X_t^n = \rho_t R_t^n \sum_{i=1}^N X_{t-1}^i, \quad t\ge2,\notag
\end{equation}
where
\begin{enumerate}[leftmargin=*,itemsep=0mm,parsep=0mm,topsep=0mm,label=(A.\arabic*)]
\item\label{assumption-growth} $\rho_t>0$ is an adapted sequence of
identically distributed random variables such that $\E (\ln \rho_t)^2 <
\infty$, $\E \ln \rho_t > 0$, and $\rho_t$ are independent of  $\F_{t-1}$ for all $t$;
\item\label{assumption-positivity} $R_t = (R_t^1,\ldots,R_t^N)$ is an
adapted sequence of random vectors with values in $\Delta$ and there
exists $\epsilon>0$ such that $\E(R_t^n\mid \F_{t-1})\ge \epsilon$ for all
$n$ and $t$.
\end{enumerate}
\end{assumption}

The sequences $\rho_t$ and $R_t$ have a rather clear interpretation. Indeed,
$\rho_t$ expresses the growth rate of the total payoff
($\sum_n X_t^n = \rho_t \sum_n X_{t-1}^n$), and $R_t^n = X_t^n/\sum_i
X_t^i$ are the relative payoffs of the assets.  Observe that if
\eqref{positive-proportions} holds, then $W_t = \sum_n X_t^n = \rho_t W_{t-1}$, and
\ref{assumption-growth} implies that $\lim_{t\to\infty}W_t=\infty$ by the
strong law of large numbers.

Introduce the following strategy $\Lambda^*$, which depends only on $t$ and
$\omega$, and has the components
\[
\Lambda_t^{*,n} \equiv \lambda^{*,n}_t = \E(R_t^n\mid \F_{t-1}).
\]
Note that this is the same strategy as $\lambda^*$ in
\cite{AmirEvstigneev+13}.

\begin{theorem}
\label{theorem-1}
Let Assumption~\ref{A} hold and suppose investor $1$ uses the strategy $
\Lambda^*$. Then $\E \tau_{l}^1 <\infty$ for any $l>0$, and for any other
investor $m\in\{2,\ldots,M\}$
\begin{equation}
\limsup_{l\to\infty} \frac{\E \tau_{l}^1}{\E \tau_l^m} \le
1.\label{minimal-time}
\end{equation}
\end{theorem}

This theorem shows that no investor can reach a wealth level $l$ faster
asymptotically (as $l\to\infty$) than an investor who uses the strategy
$\Lambda^*$. The next theorem strengthens inequality \eqref{minimal-time} if
the other investor uses an essentially different strategy. We will establish
it for the case of two investors in the market and when the following
additional assumption holds.

\begin{assumption}
\label{B} The sequence  of vectors $R_t$ from Assumption~\ref{A} is such that
\begin{enumerate}[leftmargin=*,itemsep=0mm,parsep=0mm,topsep=0mm,label=(B.\arabic*)]
\item\label{assumption-iid} $R_t$ are identically distributed and
independent of $\F_{t-1}$ for all $t$;
\item\label{assumption-linind} $R_t$ have linearly independent components,
i.e.\ if $\sum_n c^nR_t^n = 0$ a.s.\ for $c\in\R^N$, then $c=0$.
\end{enumerate}
\end{assumption}

Observe that if this assumption holds, then the strategy $\Lambda^*$ is
constant and $\lambda_t^{*,n} = \lambda^{*,n} = \E(R_1^n)>0$ for all $t$ and
$n$, where the inequality follows from~\ref{assumption-positivity}. For
$a>0$, introduce the function
\[
f(a) = \sup\Biggl\{ \E \ln\sum_{n=1}^N \frac{\lambda^n R^n_1}{ \lambda^{*,n}}
\ \Bigg|\  \lambda \in \Delta\ \text{and}\ \|\lambda-\lambda^*\| \ge a\Biggr\},
\]
where we put $f(a)=-1$ if the set under the supremum is empty. If this set
is non-empty (i.e.\ $a$ is small enough), then it is compact, so the
supremum is attained since the above expectation is upper semicontinuous in
$\lambda$ as follows from the Fatou lemma. Moreover, by Jensen's inequality,
for any $\lambda\in \Delta$, we have $\E \ln\sum_n
\lambda^nR_1^n/\lambda^{*,n} < 0$. The inequality is strict because the
logarithm is strictly concave and $\sum_n \lambda^nR_1^n/\lambda^{*,n}$ is
non-constant as follows from \ref{assumption-linind}. Consequently, $f(a) <
0$ for any $a>0$.

\begin{theorem}
\label{theorem-2}
Let $M=2$ and Assumptions~\ref{A},~\ref{B} hold. Suppose investor 1 uses the
strategy $\Lambda^*$ and investor 2 uses a strategy $\tilde\Lambda_t$ such
that its realization $\tilde\lambda_t=\tilde\lambda_t(\omega)$ satisfies the
inequality $\|\tilde\lambda_t - \lambda^*\| \ge a$ a.s.\ for all $t\ge 1$ with
some $a>0$, and $\tilde\lambda_t^n$ are uniformly bounded away from zero
(i.e.\ $\tilde\lambda_t^n > \tilde\epsilon$ for all $t,n$ and some
$\tilde\epsilon>0$).

Then, with $\theta = \E \ln \rho_1 > 0$, we have
\begin{equation}
\limsup_{l\to\infty} \frac{\E \tau_{l}^1}{\E \tau_l^2} \le 1 -
\frac{|f(a)|\wedge \theta}{\theta}.
\label{minimal-time-2}
\end{equation}
\end{theorem}

Note that the assumption about only two investors, $M=2$, is not too
restrictive. In the case $M\ge 3$, the above theorem can be used if one
replaces investors $m=2,\ldots,M$ with the \emph{representative investor}
and let $\tilde\lambda$ be the realization of the strategy of this new
investor (see
\eqref{representative-investor-1}--\eqref{representative-investor-2} below).
Since the time needed for an individual investor $m\ge 2$ to reach a given
wealth level is not less than the same time for the representative investor,
inequality \eqref{minimal-time-2} will remain valid if $\tau_l^2$ is
replaced by $\tau^m_l$.

Inequality \eqref{minimal-time} generally cannot be improved if
Assumption~\ref{assumption-linind} does not hold. In Example~\ref{example}
below, we demonstrate this for the case when $R_t$ are non-random. This fact
can be compared with the known phenomenon of \emph{volatility-induced
growth} in models with exogenous asset prices, which consists in that a
constant proportions strategy can achieve a growth rate strictly greater
than the growth rate of any asset, if the relative prices are non-constant.
If the relative prices are constant this effect disappears, which may seem
counter-intuitive since usually randomness (or volatility) is regarded as an
impediment to financial growth. A~popular intuitive explanation of this
phenomenon consists in that a constant proportions strategy ``buys low and
sells high'' (see, e.g., \cite{FernholzShay82} or Chapter~15
in~\cite{Luenberger98}), but such an explanation have known flaws
\cite{DempsterEvstigneev+07}.

\begin{example}
\label{example}
Suppose $W_0=1$ and the asset payoffs are non-random and given by
\begin{equation}
\qquad X_t^n = R^n \rho^t, 
\label{simple-model}
\end{equation}
where $\rho>1$, $R \in \Delta$ with $R^n \ge \epsilon$ for all $n$ and
some $\epsilon>0$. Clearly, this model satisfies Assumptions~\ref{A} and
\ref{assumption-iid}, and the strategy $\Lambda^*$ is of the form
$\Lambda^*_t = R$ for all $t$.

\begin{proposition}
\label{proposition-example}
Suppose in model~\eqref{simple-model} investor 1 uses the strategy
$\Lambda^*$ and investor 2 uses some constant strategy $\tilde \Lambda_t =
\tilde \lambda$. Then
\[
\lim\limits_{l\to\infty}\frac{\tau^1_l}{\tau^2_l} = 1.
\]
\end{proposition}
\end{example}

\section{Proofs}
\label{section-proofs}
We begin with two simple lemmas, which will be needed in the proofs.
Throughout this section, for a vector $x\in \R^N$ we will denote its $L^1$
and $L^2$ norms by $|x| = \sum_n |x^n|$ and $\|x\| = (\sum_n (x^n)^2)^{1/2}$
.

\begin{lemma}
\label{lemma-gibbs}
Suppose $x,y\in \R^N$ have strictly positive coordinates and $|x|=1$. Then
\begin{equation}
\sum_{n=1}^N x^n(\ln x^n - \ln y^n) \ge \frac 14 \biggl\|x-\frac{y}{|y|}\biggr\|^2 -
\ln |y|.\label{gibbs}
\end{equation}
\end{lemma}
One can see that this lemma follows from a known inequality for the
Kullback-Leibler divergence if $x$ and $y/|y|$ are considered as probability
distributions on a set of $N$ elements. A~short direct proof of
(\ref{gibbs}) can be found in \cite{AmirEvstigneev+13} (see there Lemma 2,
which is proved for $|y|=1$, but easily implies our case as well).

\begin{lemma}
\label{lemma-stopping}
Let $(\Omega,\F,(\F_t)_{t=0}^\infty,\P)$ be a filtered probability space.
Suppose $X_t$, $t\ge 1$, is an adapted sequence of identically distributed
random variables such that $X_t$ is independent of $\F_{t-1}$ and $\E X_t^2
< \infty$ for all $t\ge 1$. Denote $\mu= \E X_t$, $\sigma^2 = \Var X_t$.
Then for any stopping time $\tau\ge 1$
\[
\E X_\tau \le \mu + 2\sigma\sqrt{\E \tau}.
\]
\end{lemma}
\begin{proof}
Without loss of generality we may assume $\mu=0$ and $\E \tau < \infty$.
Introducing the martingale $M_t = \sum_{s\le t} X_s$, we obtain
\[
\E X_\tau \le \E (M_\tau - \min_{s\le \tau} M_s \I(\tau>1)) \le \|\min_{s\le
\tau} M_s\|_{L^2} \le 2 \biggl\|\Bigl(\sum_{s\le \tau}
X_s^2\Bigr)^{1/2}\biggr\|_{L^2} = 2 \sigma \sqrt{\E \tau},
\]
where in the second inequality we applied Wald's identity $\E M_\tau = 0$,
in the next one the Burkholder--Davis--Gundy inequality, and then Wald's
identity again (see, e.g., Chapter~7 of \cite{Shiryaev19} for these
results).
\end{proof}

\medskip
\noindent
\textbf{Proof of Theorem~\ref{theorem-1}}. Let us first show that the proof
can be reduced to the case when $M=2$ by replacing investors $2,\ldots,M$
with a representative investor. Let $r_t^m = Y_t^m/W_t$ denote the relative
wealth of the investors and define
\begin{equation}
\tilde Y_t = \sum_{m= 2}^M Y_t^m, \qquad \tilde \lambda_t^n = \sum_{m=2}^M
\frac{r_{t-1}^m}{1-r_{t-1}^1}
\lambda_t^{m,n},\label{representative-investor-1}
\end{equation}
where $\lambda_t^{m,n}=\lambda_t^{m,n}(\omega)$ are the realizations of the
strategies defined in \eqref{realizations}, and we put $\tilde \lambda_t^n =
0$ when $r_{t-1}^1=1$. Since $\lambda^{*,n}_t>0$ for all $n$ by
Assumption~\ref{assumption-positivity}, we have $W_t = \sum_n X_t^n>0$, so
$r_t^m$ are well-defined. Denote
\begin{equation}
\tilde \tau_l = \min\{t\ge 0: \tilde Y_t\ge l\}\label{representative-investor-1a}
\end{equation}
and observe that $\tilde \tau_l \le \tau^m_l$ for any $m\ge 2$. Also, it is
straightforward to check that the wealth sequence of investor 1 satisfies
the relation
\begin{equation}
Y_t^1 = \sum_{n=1}^N \frac{\lambda^{*,n}_t Y_{t-1}^1}{ \lambda^{*,n}_t Y_{t-1}^1 +
\tilde \lambda_t^n \tilde Y_{t-1}} X_t^n,\label{representative-investor-2}
\end{equation}
which is precisely relation \eqref{capital-equation-1} in the case of two
investors who have the wealth $Y_t^1,\tilde Y_t$ and use the strategies
$\Lambda^*_t,\tilde\Lambda_t = \tilde\lambda_t(\omega)$, while $\tilde
\tau_l$ is the first moment when the wealth of the second investor reaches
or exceeds $l$. Consequently, to prove the theorem, it would be enough to
show that
\[
\limsup_{t\to\infty} \frac{\E \tau^1_{l}}{\E \tilde \tau_l} \le 1.
\]

So, from now on we will deal with the case $M=2$. For brevity of notation,
we will denote the realizations of strategies and the wealth of the first
and the second investors, respectively, by $\lambda_t\;(=\lambda^*_t)$, $Y_t$
and $\tilde \lambda_t$, $\tilde Y_t$; their relative wealth will be denoted
by $r_t$ and $\tilde r_t=1-r_t$, and the moments of reaching or exceeding a
wealth level $l$ by $\tau_l$ and $\tilde \tau_l$.

From \eqref{representative-investor-2}, we find
\begin{equation}
\frac{r_t}{r_{t-1}} = \sum_{n=1}^N \frac{\lambda^n_t}{r_{t-1} \lambda^n_t + \tilde
r_{t-1} \tilde \lambda_t^n } R_t^n.\label{ratio}
\end{equation}
Denoting $\beta_t^n = r_{t-1}\lambda^n_t + \tilde r_{t-1} \tilde
\lambda_t^n$, we obtain the relation
\begin{equation}
\E (\ln r_t \mid \F_{t-1}) - \ln r_{t-1} \ge \E \Biggl(\sum_{n=1}^N R_t^n \ln
\frac{\lambda_t^n}{\beta_t^n}\Biggr) = \sum_{n=1}^N  \lambda^{n}_t (\ln
\lambda^n_t - \ln \beta_t^n) \ge 0,\label{submart}
\end{equation}
where in the first inequality we used the concavity of the logarithm, in the
second one that $\lambda^n_t= \E(R_t^n \mid \F_{t-1})$, and in the last
inequality applied Lemma~\ref{lemma-gibbs} to the vectors $\lambda_t$ and
$\beta_t$. Inequality \eqref{submart} implies that $\ln r_t$ is a
submartingale (the integrability of $\ln r_t$ follows from that $r_t\ge
r_{t-1} \min_n \lambda^n_t \ge \epsilon r_{t-1}$, which can be seen from
\eqref{ratio}). In passing, observe that since this submartingale is
non-positive, with probability 1 there exists the finite limit
$\lim_{t\to\infty} \ln r_t$, so $\inf_{t\ge 0} r_t > 0$. This property
allows to call $\Lambda^*$ a \emph{survival strategy}, i.e.\ an investor
``survives'' in the market by keeping a share of wealth bounded away from
zero. This result was proved in \cite{AmirEvstigneev+13} for general payoff
sequences (note that in earlier papers, e.g. \cite{BlumeEasley92}, the term
``survival'' has a somewhat different meaning).

If $\E \tilde \tau_l <\infty$, we find
\begin{equation}
\ln l \le \E \ln \tilde Y_{\tilde \tau_l} = \E(\ln
\tilde r_{\tilde \tau_l} + \ln W_{\tilde \tau_l}) \le \E \ln W_{\tilde
\tau_l} = \theta \E \tilde
\tau_{l} +  \ln  W_0,\label{bound-1}
\end{equation}
where in the last equality we applied Wald's identity to the sequence of
i.i.d.\ random variables $\ln(W_t/W_{t-1}) = \ln \rho_t$. Here, as in
Section~\ref{section-results}, $\theta = \E \ln \rho_1$.

Inequality (\ref{bound-1}) gives us the lower bound $\E \tilde \tau_l \ge
\theta^{-1} \ln (l/ W_0)$. Then we would like to obtain an upper bound for
$\E \tau_{l}$ of the same order. To do that, we will work with a slightly
altered sequence $\tilde \lambda_t$, which we will define now.

Let $\epsilon>0$ be the constant from Assumption~\ref{assumption-positivity}
and put $\delta=\epsilon^2/256$. Define recursively the sequences $r_t'$,
$t\ge 0$, and $\tilde \lambda_t'$, $t\ge 1$, by the relations
\begin{alignat}{2}
&r_0' = r_0,&& \notag\\
&\tilde \lambda_t' = \tilde\lambda_t + \delta\lambda_t\I(\min_n \tilde \lambda_t^n \le
\epsilon/2,\; r_{t-1}'\le 1/2), \qquad &&t\ge 1,\notag\\
&r_t' = \sum_{n=1}^N \frac{\lambda^n_t r_{t-1}'}{\lambda^n_t
r'_{t-1}+\tilde\lambda^{'n}_t (1-r_{t-1}')} R_t^n,
&&t\ge 1.\label{r-prime}
\end{alignat}
By induction, one can check that $r_t' \le r_t$. Put $Y_t' = r_t' W_t$ and
$\tau' = \min\{t\ge 0: Y_t' \ge l\}$. Then we have $\tau_l \le \tau'$, so we
will look for an upper bound for $\E \tau'$.

Similarly to (\ref{submart}), we can show that $\ln r_t'$ is a
submartingale. Indeed, let $\beta'_t = r_{t-1}' \lambda_t + (1-r_{t-1}')
\tilde \lambda_t'$. Then
\begin{equation}
\E (\ln r_t' \mid \F_{t-1}) - \ln r_{t-1}' \ge \sum_{n=1}^N  \lambda^{n}_t (\ln
\lambda^n_t - \ln {\beta_t'}^n) \ge \frac14 \biggl\| \lambda_t - \frac{\beta_t'}{|\beta_t'|} \biggr\|^2 -
\ln|\beta_t'| .\label{submart-2}
\end{equation}
On the event $\{\tilde\lambda'_t = \tilde\lambda_t\}$ we have
$|\beta_t'|=1$, so the right-hand side of \eqref{submart-2} is non-negative.
On the event $\{\tilde\lambda'_t = \tilde\lambda_t + \delta\lambda_t\}$,
there exists a coordinate $n = n(\omega)$ such that $\tilde \lambda_t^n \le
\epsilon/2$, so, using that $|\beta'_t| = 1 + \delta(1-r'_{t-1}) \le
1+\delta$, we can estimate $\ln| \beta_t'| \le \delta$ and
\[
\biggl\|\lambda_t - \frac{\beta_t'}{|\beta_t'|} \biggr\| \ge
\frac{1-r_{t-1}'}{|\beta_t'|} ( \lambda^n_t - \tilde \lambda_t^n ) \ge
\frac\epsilon8.
\]
Then the choice of $\delta$ implies that the right-hand side of
(\ref{submart-2}) is non-negative on the event $\{\tilde\lambda'_t =
\tilde\lambda_t + \delta\lambda_t\}$ as well. Thus, $\ln r_t'$ is a non-positive
submartingale, and, in particular, $\inf_{t\ge 0} r'_t> 0$. Since
$W_t\to\infty$, we also have $\tau'<\infty$.

From now on, assume that $l > Y_0$ (since we take $l\to\infty$). Applying
Fatou's lemma, we obtain
\[
\ln l \ge \limsup_{t\to\infty }\E \ln Y'_{\tau'\wedge t-1} =
\limsup_{t\to\infty } \E \biggl(\ln W_{\tau'\wedge t} - \ln \rho_{\tau'\wedge
t} + \ln r'_{\tau'\wedge t} - \ln \frac{r'_{\tau'\wedge t}}{r'_{\tau'\wedge
t-1}}\biggr).
\]
By Wald's identity, $\E \ln W_{\tau'\wedge t} = \theta \E (\tau'\wedge t) +
\ln W_0$. From Lemma~\ref{lemma-stopping}, $\E \ln \rho_{\tau'\wedge t} \le
\theta + 2 \sigma \sqrt{\E (\tau'\wedge t)}$, where $\sigma^2 = \Var(\ln
\rho_1)$. Since $\ln r_t'$ is a submartingale, $\E \ln {r'_{\tau'\wedge t}
\ge \ln r_0}$. Finally, for all $t\ge 1$ we have $r_t' \le
r'_{t-1}/(\delta\epsilon)$. Indeed, if $r_{t-1}'>1/2$ this is obvious, while
if $r_{t-1}'\le 1/2$ we can use \eqref{r-prime} and
the inequalities $\lambda_t^n \ge \epsilon\ge \delta\epsilon$ and $\tilde
\lambda_t^{'n} \ge \min(\epsilon/2,\; \delta\lambda_t^n) \ge \delta\epsilon$
to find
\[
r_t' \le \sum_{n=1}^N \frac{r_{t-1}'}{\lambda_t^n r_{t-1}' +
\tilde \lambda_t^{'n}(1-r_{t-1}')} R_t^n \le \frac{ r_{t-1}'}{\delta\epsilon}.
\]
Consequently, we obtain the inequality
\[
\ln l > \limsup_{t\to\infty} \Bigl(\theta \E (\tau'\wedge t) -
2\sigma\sqrt{\E(\tau'\wedge t)}\Bigr) - \theta + \ln W_0 + \ln r_0 + \ln (\delta\epsilon).
\]
Applying the monotone convergence theorem, we can see that $\E \tau'$ should
be finite, and hence
\begin{equation}
\ln l \ge \theta \E \tau' - 2\sigma \sqrt{\E \tau'} - \theta + \ln
(Y_0\delta\epsilon).\label{bound-2}
\end{equation}
Now the claim of the theorem follows from \eqref{bound-1}, \eqref{bound-2},
and the relation $\E \tau_l \le \E \tau'$.
\qed

\bigskip
In the following proofs of Theorem~\ref{theorem-2} and
Proposition~\ref{proposition-example}, we will use the same notation for
investors 1 and 2 as in the proof of Theorem~\ref{theorem-1}, i.e.\ without
and with tilde, respectively.

\bigskip
\noindent
\textbf{Proof of Theorem~\ref{theorem-2}}. We can assume that $\E \tilde
\tau_l < \infty$ for all $l>0$, as otherwise the proof becomes trivial. Let
us begin with an auxiliary estimate. For $c\in[0,1)$ we define $\eta_c =
\sum_{t\ge 0} \I(r_t< c)$ and will now show that $\E \eta_c < \infty$. As
was shown in the proof of Theorem~\ref{theorem-1}, $\ln r_t$ is a
non-positive submartingale. If we denote by $C_t$ its compensator, i.e.\ the
non-negative and non-decreasing sequence
\[
C_t := \sum_{s=1}^t\E\biggl(\ln \frac{r_{s}}{r_{s-1}} \;\bigg|\; \F_{s-1}\biggr), 
\]
then $C_t$ a.s.-converges to a limit $C_\infty$ with $\E C_\infty < \infty$.
This follows from the monotone convergence theorem since $\E C_t = \E \ln
(r_t/r_0) \le -\ln r_0$. Using Lemma~\ref{lemma-gibbs}, similarly to
\eqref{submart} and \eqref{submart-2}, we see that (with the same $\beta_t$
as in \eqref{submart})
\begin{equation}
C_\infty \ge \frac14 \sum_{t=1}^\infty \|\lambda - \beta_t\|^2 = \frac14
\sum_{t=1}^\infty (1-r_{t-1})^2\|\lambda - \tilde\lambda_t\|^2 \ge \frac{a^2}{4}
\sum_{t=1}^\infty (1-r_{t-1})^2.\label{compensator-bound}
\end{equation}
Therefore,  $\eta_c \le 4 C_\infty/(a(1-c))^2$, so $\E \eta_c <
\infty$.

From \eqref{ratio}, we find
\[
\frac{\tilde r_t}{\tilde r_{t-1}} = \sum_{n=1}^N \frac{\tilde \lambda^{n}_t}{
\lambda^n r_{t-1} +  \tilde \lambda_t^{n} \tilde r_{t-1}} R_t^n,
\]
which implies 
\[
\E \biggl(\ln \frac{\tilde r_t}{\tilde r_{t-1}} \;\bigg|\; \F_{t-1}\biggr)
\le \E\biggl( \ln \sum_{n=1}^N \frac{\tilde \lambda^{n}_t R_t^n}{  \lambda^n
r_{t-1}}  \;\bigg|\; \F_{t-1}\biggr) \le f(a) - \ln r_{t-1}.
\]
Since $r_t$ is a submartingale, we have $\E (\ln (\tilde r_t/\tilde r_{t-1})
\mid \F_{t-1}) \le 0$ by Jensen's inequality, so $\ln \tilde r_t$ is a
supermartingale (its integrability follows from that $\tilde\epsilon \le
\tilde r_t/\tilde r_{t-1} \le\min(\epsilon,\tilde\epsilon)^{-1}$).
Consequently, using that $\E \tilde \tau_l<\infty$ and applying Doob's
stopping theorem, we obtain
\begin{equation}
\E \ln \tilde r_{\tilde \tau_l} \le \E \sum_{s=1}^{\tilde \tau_l}\min(f(a) -
\ln r_{s-1}, 0).
\label{rtilde-bound}
\end{equation}
The possibility of applying Doob's theorem can be justified by first
applying it to the bounded stopping times $\tilde \tau_l \wedge t$, then
passing to the limit $t\to\infty$ using Fatou's lemma in the left-hand side
of \eqref{rtilde-bound} (note that $\ln \tilde r_{\tilde \tau_l\wedge t}$ is
bounded from below by the integrable random variable $ \tilde \tau_l \ln
\tilde \epsilon + \ln\tilde r_0$), and using the monotone convergence
theorem in the right-hand side.

Now, similarly to \eqref{bound-1}, for any $c\in [e^{f(a)},1)$ we find
\[
\begin{split}
\ln \frac{l}{W_0} &\le \E \ln \tilde r_{\tilde \tau_l} + \theta \E \tilde \tau_l
\le \E\sum_{s=1}^{\tilde \tau_l} (f(a) - \ln(c)) I(r_{s-1}\ge c) + \theta
\E \tilde \tau_l \\ &\le (\theta+f(a) - \ln(c)) \E \tilde \tau_l +
(\ln c - f(a)) \E \eta_c.
\end{split}
\]
Note that since we consider the case $\E \tilde\tau_l <\infty$ for all
$l>0$, we necessarily have $\theta + f(a) - \ln(c) > 0$. Together with
\eqref{bound-2}, this implies
\[
\limsup_{l\to\infty} \frac{\E  \tau_{l}}{\E \tilde \tau_{l}} \le \frac{\theta +
f(a) - \ln c}{\theta}.
\]
Taking $c\to 1$, we obtain the claim of the theorem. \qed

\bigskip
\noindent
\textbf{Proof of Proposition~\ref{proposition-example}.} We will assume that
$\tilde\lambda \neq \lambda$, as otherwise the claim of the proposition is
obvious. Since $W_t = \rho ^t$, for investor 1 we have $\tau_l \ge
\theta^{-1}\ln l$, where $\theta = \ln \rho$. Therefore, it will be enough
to show that for investor 2 we have
\begin{equation}
\tilde \tau_l \le \frac{\ln l}{\theta}(1 + o(1)).\label{prop1-bound}
\end{equation}
Using the wealth equation \eqref{capital-equation-1} and that $\lambda^n =
R^n$, we obtain
\begin{equation}
\frac{\tilde Y_t}{\tilde Y_{t-1}} = \rho \sum_{n=1}^N \frac{\tilde \lambda^n
\lambda^n}{ \lambda^n r_{t-1} + \tilde \lambda^n \tilde r_{t-1}}.
\label{nonrandom-wealth}
\end{equation}
Inequality \eqref{compensator-bound} implies that $ r_t\to 1$, so the
right-hand side of \eqref{nonrandom-wealth} is strictly greater than~$1$ for
$t$ large enough. Hence $\tilde Y_t \to \infty$, which implies that $\tilde
\tau_l <\infty$ for all $l$. Consequently,
\begin{equation}
\ln l> \ln \tilde Y_{\tilde\tau_l-1} = \ln W_{\tilde \tau_l-1} + \ln
\tilde r_{\tilde \tau_l-1} =  \theta(\tilde \tau_l-1) + \ln \tilde
r_{\tilde \tau_l-1}.\label{lnl-bound}
\end{equation}
From \eqref{nonrandom-wealth}, using the concavity of the logarithm and the
inequality $\ln x \ge 1-x^{-1}$, we obtain the bound
\[
\ln \frac{\tilde r_t}{\tilde r_{t-1}} \ge \sum_{n=1}^N \tilde \lambda^n \ln
\frac{\lambda^n}{\lambda^n r_{t-1} + \tilde \lambda^n \tilde r_{t-1}} \ge
\ln r_{t-1}^{-1} - \sum_{n=1}^N \frac{ (\tilde\lambda^n)^2\tilde r_{t-1}}{ \lambda^n
r_{t-1}} \ge \tilde r_{t-1} \biggl(1- \sum_{n=1}^N \frac{(\tilde
\lambda^n)^2}{\lambda^n r_0}\biggr),
\]
where in the last inequality we estimated $r_{t-1} \ge r_0$ since $r_t$ is a
non-decreasing sequence (in the proof of Theorem~\ref{theorem-1}, we showed
that it is a submartingale). Since $\tilde r_t\to 0$ and $\tilde
\tau_l\to\infty$, we get $\tilde \tau_l^{-1}{\ln \tilde r_{\tilde \tau_l}}
\to 0$. Then relation \eqref{lnl-bound} implies \eqref{prop1-bound}, which
is what is needed.

\phantomsection
\addcontentsline{toc}{chapter}{\refname}
\bibliographystyle{abbrvnat}
\bibliography{wealth-level}
\end{document}